\documentclass[preprint, 12pt]{elsarticle}
\usepackage{pgfplots,bm} 
\usepackage{url} 
\usepackage{amsmath, amssymb} 

\newcommand{\px}{\partial_x}
\newcommand{\py}{\partial_y}
\newcommand{\pz}{\partial_z}
\newcommand{\pxx}{\partial^2_{xx}}
\newcommand{\pxy}{\partial^2_{xy}}
\newcommand{\pyy}{\partial^2_{yy}}
\newcommand{\pyz}{\partial^2_{yz}}
\newcommand{\pxz}{\partial^2_{xz}}
\newcommand{\pzz}{\partial^2_{zz}}
\newcommand{\go}{\nabla_0}
\newcommand{\gp}{\nabla_p}
\newcommand{\gm}{\nabla_m} 
\newcommand{\br}{{\bm{r}}} 
\newcommand{\bx}{\bm{r}_i}

\newcommand{\by}{{\bm{r}_j}}

\newcommand{\rmn}{|\bx-\by|}

\newcommand{\conj}[1]{\overline{#1}}
\newtheorem{theorem}{Theorem}
\newtheorem{proof}{Proof}

\allowdisplaybreaks
\begin{document}

\begin{frontmatter}
\title{RPYFMM: Parallel Adaptive Fast Multipole Method for 
Rotne-Prager-Yamakawa Tensor in Biomolecular Hydrodynamics Simulations}

\author[a]{W. Guan}
\author[b]{X. Cheng}
\author[a]{J. Huang}
\author[e]{G. Huber}
\author[d]{W. Li}
\author[e]{J. A. McCammon}
\author[f]{B. Zhang\corref{author}}

\cortext[author]{
Corresponding author.
\textit{E-mail address:} zhang416@indiana.edu
}
\address[a]{
  Department of Mathematics, University of North Carolina, Chapel
  Hill, NC 27599-3250, USA
}
\address[b]{
  Medicinal Chemistry and Pharmacognosy, College of Pharmacy, Ohio
  State University, Columbus, OH 43210, USA
}
\address[d]{
  School of Transportation and Vehicle Engineering, Shandong
  University of Technology, Zibo, Shandong, China
}
\address[e]{
  Department of Chemistry and Biochemistry and Department of
  Pharmacology, University of California at San Diego, La Jolla, CA
  92093-0365.
}
\address[f]{
  Center for Research in Extreme Scale Technologies, School of
  Informatics, Computing, and Engineering, Indiana University,
  Bloomington, IN, 47404, USA
}

\begin{abstract}
RPYFMM is a software package for the efficient evaluation of the potential field
governed by the Rotne-Prager-Yamakawa (RPY) tensor interactions in biomolecular
hydrodynamics simulations. In our algorithm, the RPY tensor is decomposed as 
a linear combination of four Laplace interactions, each of which is evaluated using 
the adaptive fast multipole method (FMM) \cite{greengard1997new} 
where the exponential expansions are applied to
diagonalize the multipole-to-local translation operators. RPYFMM offers a unified
execution on both shared and distributed memory computers by leveraging the
DASHMM library~\cite{zhang2016dashmm, dashmmrev}. Preliminary numerical results
show that the interactions for a molecular system of $15$ million particles (beads) can
be computed within one second on a Cray XC30 cluster using $12, 288$ cores,
while achieving approximately $54\%$ strong-scaling efficiency.
\end{abstract}

\begin{keyword}
Brownian dynamics;
Rotne-Prager-Yamakawa tensor;
Hydrodynamics interactions;
Fast multipole method;
DASHMM;
Distributed computing
\end{keyword}

\end{frontmatter}

{\bf \noindent PROGRAM SUMMARY}

\begin{small}
\noindent
{\em Program Title:} RPYFMM: Parallel Adaptive FMM for RPY Tensor \\
{\em Licensing provisions:} GNU General Public License, version 3\\
{\em Programming language:} C++ \\
{\em Supplementary material:} \\
{\em Nature of problem:} Evaluate the Rotne-Prager-Yamakawa tensor matrix-vector
multiplications describing the hydrodynamics interaction in biomolecular
systems. \\
{\em Solution method:} The Rotne-Prager-Yamakawa tensor is decomposed as a linear
combination of four Laplace interactions, each of which is evaluated using the new
version of adaptive fast multipole method~\cite{greengard1997new}.  \\
{\em Additional Comments:} RPYFMM is built on top of the DASHMM library and the
Asynchronous Multi-Tasking HPX-5 runtime system. DASHMM is automatically
downloaded during installation and HPX-5 is available at
\url{http://hpx.crest.iu.edu/}.

\end{small}

\section{Introduction}
\label{sec:intro}
The dynamics of macromolecules inside a living cell usually takes place at very
low Reynolds numbers, where the viscous forces dominate over inertial effects,
and the {\it in vivo} macromolecule diffusion is strongly influenced by the
hydrodynamic interactions (HIs). Due to its long-range nature and many-body
character, HI is responsible for a wide variety of fascinating collective
phenomena. Depending on whether HI is present or not, existing studies have
shown qualitative differences in the aggregation or microphase separation of
colloids \cite{Tanaka2000}. In particular, HI was shown to facilitate barrier
crossing during the microphase separation of block copolymers, whereas without
HI, the system appeared to be trapped in a metastable state.  HI was also found
to greatly accelerate the kinetics of lipid membrane self-assembly
\cite{Ando2013}.

There have been many research efforts to develop accurate mathematical models
and efficient simulation tools for understanding the HI effects on biomolecular
dynamics. In this paper, we focus on the bead model \cite{wang2013assessing}, a
particular realization of the commonly used Ermak-McCammon model
\cite{Ermak1978}. In this model, the biomolecular system is treated as a system
of $N$ beads, each representing a Brownian molecule, the translational and
rotational displacement $\Delta \bm{X}_i$ of the $i$th bead during time step
$\Delta t$, due to external force $F_j$ acting on bead $j$, $j=1,\ldots, N$, and
the random displacement $R_i$, is given by
\begin{equation} 
\label{eq:Ermak}
\Delta \bm{X}_i\left( \Delta t \right)  
 = \sum_{j}{}\frac{D_{ij}F_{j}}{k_{B}T}\Delta t 
 \, + \, \sum_{j}{}\frac{\partial D_{ij}}{\partial r_j}\Delta t 
 \, + \, R_i\left(\Delta t\right), 
\end{equation}
where $k_B$ is the Boltzmann constant and $T$ is the absolute temperature.  The
external forces $F$ include the electrostatic and van der Waals
interactions. The hydrodynamic forces are accounted for by the $6N \times 6N$
diffusion tensor $D$ that describes the hydrodynamic coupling between the $N$
beads with their three translational and three rotational degrees of
freedom. The random displacements in $R$ have mean zero and variance matrix
$2D \Delta t$. When the rotational motions of the $N$ beads of radii $a$ are neglected,
the dimension of $D$ is reduced to $3N \times 3N$ and the most common form is
given by the Rotne-Prager-Yamakawa (RPY) tensor
\cite{batchelor1976brownian,Ermak1978} as follows:
\begin{align} 
D_{ii}(\br_i,\br_i) & = \frac{k_BT}{6\pi\eta a}I,\\
D_{ij}(\br_i,\br_j) & = \frac{k_BT}{8\pi\eta r_{ij}} \left [ 
(I + \frac{\br_{ij} \otimes \br_{ij}}{r_{ij}^2}) + 
  \frac{2a^2}{3r_{ij}^2} (I - \frac{3 \br_{ij} \otimes \br_{ij}}{r_{ij}^2})
  \right], \quad i\neq j, r_{ij} \geq 2a \label{eq:rpy-far}\\
D_{ij}(\br_i,\br_j) & = \frac{k_BT}{6\pi\eta r_{ij}} \left [
(1 - \frac{9r_{ij}}{32a}) I + \frac{3}{32a} \frac{\br_{ij} \otimes
    \br_{ij}}{r_{ij}^2} \right], \quad i \neq j, r_{ij} < 2a 
\end{align}
Here $\eta$ is the solvent viscosity, $i$ and $j$ label bead indices, $I$ is the
$3\times 3$ identity matrix, $\br_i = [x_i, y_i, z_i]^T$ is the $3\times 1$
position vector of bead $i$, $\br_{ij} = \br_j - \br_i$, and $r_{ij}
=\|\br_{ij}\|_2$ is the distance between beads $i$ and $j$. In this paper, we
focus on the positive definite RPY tensor. The generalized RPY tensor with both
translational and rotational motions was presented in \cite{wajnryb2013generalization} 
and its efficient evaluation is still an active research topic.

There are several numerical difficulties in solving (\ref{eq:Ermak})
efficiently, including the evaluation of the electrostatic force field
contributing to the external forces $F$, and the algebraic operations on the
dense diffusion matrix $D$. For instance, direct calculation of all the two-body
HI interactions requires $O(N^2)$ operations, and computing hydrodynamically
correlated random displacement vectors $R$ requires prohibitive $O(N^3)$
operations via the Cholesky factorization. For these reasons, in most previous
studies, HI has been either completely neglected or considered only for a much
smaller equivalent sphere system, where the detailed molecular shape was
ignored, and each protein---a large set of beads---was modeled by an equivalent
sphere of the same hydrodynamic radius.

There exist many research efforts aimed at reducing the computational complexity
of solving (\ref{eq:Ermak}), either for the steady state or dynamic settings. A
few representative work include the parallel adaptive fast multipole method for
evaluating the electrostatic potential modeled by the linearized
Poisson-Boltzmann equation on distributed memory computers~\cite{BoAFMPB}, the
Particle-Mesh Ewald (PME) summation method for the matrix-vector multiplication
($DF$) which scales as $O(N\log N)$ \cite{Liu2014}, and different techniques for
the efficient generation of the random vector $R$ in
\cite{Fixman1986,Banchio2003,geyer,Ando2012}.

The purpose of this paper is to present a parallel software package for the
efficient evaluation of the HI interactions ($D\cdot F$) modeled by the RPY
tensor. The package is also essential to the efficient generation of the random
vectors $R$ by approximating $\sqrt{D}v$ with \{$D^kv$\}, $k=0,1,2, \cdots, k$
via the preconditioned Krylov subspace iterations~\cite{liang2013fast}.  
Our numerical algorithm uses the technique introduced in \cite{tornberg2008fast} 
for Stokeslet and \cite{jiang2013fast} for RPY tensor to decompose the RPY tensor 
as a linear combination of four Laplace potentials and their derivatives. 
Each Laplace potential is evaluated using the adaptive new version fast multipole
method~\cite{greengard1997new} where the exponential expansion is used to
diagonalize the multipole-to-local translation operators. The package is built 
on top of the open-source DASHMM library~\cite{zhang2016dashmm, dashmmrev} 
developed by some of the authors, and the Asynchronous
Multi-Tasking HPX-5 runtime, providing a unified execution on both shared and
distributed memory computers. Preliminary numerical results show that for a
molecular system with $15$ million beads, the package is able to compute the HI
within {\bf one} second on a Cray XC30 cluster using $12, 288$ cores and
achieves $54\%$ strong-scaling efficiency.

This paper is organized as follows. Section~\ref{sec:bg} reviews the
mathematical foundations of RPYFMM, including the decomposition of the RPY
tensor and how to compute the gradient and Hessian of the Laplace
potentials. Section~\ref{sec:cs} describes the main components of HPX-5 runtime
system and parallelization strategy of DASHMM library.
Section~\ref{sec:install} provides the installation guide and job examples, and
demonstrates the strong-scaling efficiency for different accuracy
requirements. Section~\ref{sec:conclusion} concludes the paper by discussing
several related research topics in order to build the next generation of
Brownian dynamics simulation package.

\section{Mathematical foundations of RPYFMM} 
\label{sec:bg} 
Similar to the electrostatic interaction, the hydrodynamics interaction (HI) modeled 
by the RPY tensor decays like $O(1/|\bx-\by|)$ as $\|\bx - \by\|_2 \to \infty$ 
and is therefore considered long range. To efficiently evaluate
these long-range hydrodynamics interactions, we apply the technique first
proposed for the Stokeslet~\cite{tornberg2008fast} and later generalized to the
RPY tensor~\cite{jiang2013fast} to decompose the far-field RPY tensor
(\ref{eq:rpy-far}) as a linear combination of four scalar Laplace potentials and
their derivatives. For a target bead $i$ located at
$\bx=[x_i, y_i, z_i]^T$, denote the set of well-separated beads by
$\Omega_i^{far}$, where each bead $j \in \Omega_i^{far}$ is located at
$\by=[x_j, y_j, z_j]^T$ exerting force $F_j = [F_j^1, F_j^2, F_j^3]^T$.
The four Laplace potentials are defined as follows:
\begin{align} 
L_1(\bx) & = \sum_{j \in \Omega_i^{far}} \frac{F_j^1}{\rmn} \\
L_2(\bx) & = \sum_{j \in \Omega_i^{far}} \frac{F_j^2}{\rmn} \\
L_3(\bx) & = \sum_{j \in \Omega_i^{far}} \frac{F_j^3}{\rmn} \\
L_4(\bx) & = \sum_{j \in \Omega_i^{far}} 
\frac{F_j^1 x_j + F_j^2 y_j + F_j^3  z_j}{\rmn} 
= \sum_{j \in \Omega_i^{far}} \frac{F_j \cdot \by}{\rmn} 
\end{align} 
To simplify the notations later, we further define 
\begin{equation}
  L_C(\bx) = \frac{\partial L_1(\bx)}{\partial x_i} +
  \frac{\partial L_2(\bx)}{\partial y_i} 
  + \frac{\partial L_3(\bx)}{\partial z_i}. 
\end{equation}
Using these notations, the far-field HI at target $i$ due to contributions from
$\Omega_i^{far}$ can be collected as
\begin{align} 
\phi_{far} (\bx) &= 
[u_i, v_i, w_i]^T = \sum_{j \in \Omega_i^{far}} D_{ij} (\bx, \by) F_j \nonumber \\
& = \begin{bmatrix} 
C_1 L_1 - C_1 \left(x_i \frac{\partial L_1}{\partial x_i} 
 + y_i \frac{\partial L_2}{\partial x_i} 
 + z_i \frac{\partial L_3}{\partial x_i} \right) + 
C_1 \frac{\partial L_4}{\partial x_i} + 
C_2 \frac{\partial L_C}{\partial x_i} \\
C_1 L_2 - C_1 \left(x_i \frac{\partial L_1}{\partial y_i} 
 + y_i \frac{\partial L_2}{\partial y_i} 
 + z_i \frac{\partial L_3}{\partial y_i} \right) + 
C_1 \frac{\partial L_4}{\partial y_i} + 
C_2 \frac{\partial L_C}{\partial y_i} \\
C_1 L_3 - C_1 \left(x_i \frac{\partial L_1}{\partial z_i} 
 + y_i \frac{\partial L_2}{\partial z_i} 
 + z_i \frac{\partial L_3}{\partial z_i} \right) + 
C_1 \frac{\partial L_4}{\partial z_i} + 
C_2 \frac{\partial L_C}{\partial z_i} \\
\end{bmatrix} \nonumber \\
& = C_1 \begin{bmatrix} L_1 \\ L_2 \\ L_3 \end{bmatrix} - 
C_1 (x_i \nabla L_1 + y_i \nabla L_2 + z_i \nabla L_3)+ 
C_1 \nabla L_4 + C_2 \nabla L_C \label{eq2.8} 
\end{align} 
where $C_1=\frac{k_BT}{8\pi\eta}$ and $C_2=\frac{k_BTa^2}{12\pi\eta}$.
Notice that to compute the HI in (\ref{eq2.8}), one has to compute the
three Laplace potentials $L_1$, $L_2$, $L_3$, the gradients $\nabla
L_1$, $\nabla L_2$, $\nabla L_3$, $\nabla L_4$, and the Hessians of
$L_1$, $L_2$, and $L_3$ implicitly expressed in $\nabla L_C$.

In the FMM, far-field potentials are collected in the form of {\it multipole} or
{\it local} expansions. In RPYFMM, the multipole ($\mathbf{M}$) and local
($\mathbf{L}$) expansions for the Laplace potential are of the form 
\[
\mathbf{M} = 
\sum_{n=0}^p\sum_{m=-n}^n \frac{M_n^m Y_n^m}{r^{n+1}} \quad 
\text{and} \quad 
\mathbf{L} = \sum_{n=0}^p \sum_{m=-n}^n L_n^m Y_n^m r^n, 
\]
where the spherical harmonics $Y_n^m$ is defined as 
\[
Y_n^m (\theta, \phi)  = \sqrt{\frac{(n - |m|)!}{(n + |m|)!}}
P_n^{|m|}(\cos \theta) e^{im \phi} = C_n^m P_n^{|m|} e^{im\phi} 
\]
Under this definition, it is easy to verify that $M_n^{-m} = \conj{M_n^m}$ and
$L_n^{-m} = \conj{L_n^m}$ and one only saves the coefficients for $0 \leq m \leq
n$ in the implementation.

To compute the gradient and Hessian of the above multipole (local)
expansion, one defines operators $\go=\pz$, $\gp=\px+i\py$, and
$\gm=\px-i\py$. If $\phi$ is a harmonic function, then 
\[
\gm\gp \phi= \gp\gm \phi = -\go\go\phi.
\]
When the spherical harmonics follow the conventional definition
\[
Y_n^m(\theta, \phi) = \sqrt{\frac{2n+1}{4\pi}}
\sqrt{\frac{(n - m)!}{(n + m)!}} P_n^m (\cos\theta) e^{im\phi}, \quad
Y_n^{-m} = (-1)^m \conj{Y_n^m}, \quad 0 \leq m \leq n,
\]
the following relations hold 
\begin{align}
\go\left(\frac{Y_n^m}{r^{n+1}}\right) & = 
- \sqrt{\frac{(n + m + 1)(n - m + 1)}{(2n + 1)(2n + 3)}} 
(2n+1) \frac{Y_{n+1}^m}{r^{n+2}} \\
\gp\left(\frac{Y_n^m}{r^{n+1}}\right) & = 
\sqrt{\frac{(n + m + 2)(n + m + 1)}{(2n+1)(2n+3)}} 
(2n+1) \frac{Y_{n+1}^{m+1}}{r^{n+2}} \\
\gm\left(\frac{Y_n^m}{r^{n+1}}\right) & = 
-\sqrt{\frac{(n - m + 2)(n - m + 1)}{(2n + 1)(2n+3)}}
(2n+1) \frac{Y_{n+1}^{m-1}}{r^{n+2}}\\
\go\left(Y_n^m r^n\right) & = 
\sqrt{\frac{(n+m)(n - m)}{(2n-1)(2n+1)}}(2n+1)r^{n-1}Y_{n-1}^m \\
\gp\left(Y_n^m r^n\right) & = 
\sqrt{\frac{(n - m - 1)(n - m)}{(2n-1)(2n+1)}} (2n+1)r^{n-1}Y_{n-1}^{m+1} \\
\gm\left(Y_n^m r^n\right) & = 
-\sqrt{\frac{(n+m-1)(n+m)}{(2n-1)(2n+1)}} (2n+1) r^{n-1} Y_{n-1}^{m-1}. 
\end{align}
As the Laplace potential is real, the gradient and Hessian of the
multipole (local) expansion can be obtained
\begin{align}
  & \pz \leftarrow \go, \quad \px \leftarrow \Re{\gp}, \quad \py
  \leftarrow \Im{\gp} \nonumber \\
  & \pxz \leftarrow \Re{\gp\go}, \quad 
  \pyz \leftarrow \Im{\gp \go}, \quad
  \pzz \leftarrow \go\go, \quad 
  \pxy \leftarrow \Im{\gp\gp}/2\nonumber \\
  & 
  \pxx \leftarrow \Re{(\gp\gp - \go\go)}/2, \quad 
  \pyy \leftarrow \Re{( -\gp\gp - \go\go)}/2
\end{align} 

\begin{theorem}
Let $\mathbf{M} = \sum_{n=0}^p \sum_{m=-n}^n 
\frac{M_n^m C_n^m P_n^m e^{im\phi}}{r^{n+1}}$ be the multipole expansion. Then 
\begin{align*}
& \go \mathbf{M} = s_1 + 2 \Re{s_2}, \quad \gp \mathbf{M} = s_3 + \conj{s_4}, \\
& \gp\go \mathbf{M} = s_5+ \conj{s_6}, \quad 
\gp\gp \mathbf{M} = s_7+ \conj{s_8}, \quad  
\go\go \mathbf{M} = s_{9} + 2\Re{s_{10}}, 
\end{align*}
where
\begin{align*}
s_1 & = - \sum_{n=0}^p \frac{M_n^0 (n+1) P_{n+1}^0}{r^{n+2}} \\
s_2 & = - \sum_{n=1}^p \sum_{m=1}^n \frac{M_n^m P_{n+1}^m e^{im\phi}}{r^{n+2}} 
\sqrt{\frac{(n-m+1)!(n-m+1)!}{(n-m)!(n+m)!}} \\
s_3 & = \sum_{n=0}^p \sum_{m=0}^n \frac{M_n^m P_{n+1}^{m+1}e^{i(m+1)\phi}}{r^{n+2}} 
\sqrt{\frac{(n-m)!}{(n+m)!}} \\
s_4 & = -\sum_{n=1}^p \sum_{m=1}^n \frac{M_n^m P_{n+1}^{m-1}e^{i(m-1)\phi}}{r^{n+2}}
\sqrt{\frac{(n-m+2)!(n-m+2)!}{(n-m)!(n+m)!}}\\
s_5 & = -\sum_{n=0}^p \sum_{m=0}^n 
\frac{M_n^m P_{n+2}^{m+1} e^{i(m+1)\phi}}{r^{n+3}} 
\sqrt{\frac{(n-m+1)!(n-m+1)!}{(n-m)!(n+m)!}}\\
s_6 & = \sum_{n=1}^p \sum_{m=1}^n \frac{M_n^mP_{n+2}^{m-1}e^{i(m-1)\phi}}{r^{n+3}} 
\sqrt{\frac{(n-m+3)!(n-m+3)!}{(n-m)!(n+m)!}} \\
s_7 & = \sum_{n=0}^p \sum_{m=0}^n 
\frac{M_n^m P_{n+2}^{m+2} e^{i(m+2)\phi}}{r^{n+3}}
\sqrt{\frac{(n-m)!}{(n+m)!}}\\
s_8 & = \sum_{n=2}^p\sum_{m=2}^n 
\frac{M_n^m P_{n+2}^{m-2}e^{i(m-2)\phi}}{r^{n+3}} 
\sqrt{\frac{(n-m+4)!(n-m+4)!}{(n-m)!(n+m)!}}  \\
& \qquad - \sum_{n=1}^p \frac{M_n^1 P_{n+2}^1 e^{-i\phi}}{r^{n+3}} 
\sqrt{\frac{(n+1)!}{(n-1)!}}\\
s_{9} & = \sum_{n=0}^p \frac{M_n^0 P_{n+2}^0 (n+1)(n+2)}{r^{n+3}}\\
s_{10} & = \sum_{n=1}^p \sum_{m=1}^n \frac{M_n^m P_{n+2}^m e^{im\phi}}{r^{n+3}}
\sqrt{\frac{(n-m+2)!(n-m+2)!}{(n-m)!(n+m)!}}.
\end{align*}
\end{theorem}

\begin{proof}
We show the result for $\gp\go$. Apply $\gp\go$ on each term of the
multipole expansion. When $m \geq 0$, 
\begin{align*}
& \gp \go \left(\frac{M_n^m C_n^m P_n^m e^{im\phi}}{r^{n+1}} \right) \\
= & -\gp \left (M_n^m \sqrt{\frac{4\pi}{2n+1}}
\sqrt{\frac{(n+m+1)(n-m+1)}{(2n+1)(2n+3)}} (2n+1) \frac{Y_{n+1}^m}{r^{n+2}} 
\right) \\
= & - \frac{M_n^m P_{n+2}^{m+1} e^{i(m+1)\phi}}{r^{n+3}} 
\sqrt{\frac{(n-m+1)!(n-m+1)!}{(n-m)!(n+m)!}}
\end{align*}
When $m \geq 1$, 
\[
\gp \go (\frac{\conj{M_n^m} C_n^m P_n^m e^{-im\phi}}{r^{n+1}})
= \gp \conj{\go \frac{{M_n^m} C_n^m P_n^m e^{im\phi}}{r^{n+1}}}
= \conj{\gm \go \frac{{M_n^m} C_n^m P_n^m e^{im\phi}}{r^{n+1}}}
\]
Similar algebraic work gives 
\[
\gm\go (\frac{M_n^m C_n^m P_n^m e^{im\phi}}{r^{n+1}})
=  \frac{M_n^mP_{n+2}^{m-1}e^{i(m-1)\phi}}{r^{n+3}} 
\sqrt{\frac{(n-m+3)!(n-m+3)!}{(n-m)!(n+m)!}} 
\]
The results for the other operators can be obtained with similar 
algebraic work. 
\end{proof} 

One can similarly obtain the following result for the local
expansion. 
\begin{theorem}
Let $\mathbf{L} = \sum_{n=0}^p \sum_{m=-n}^n L_n^m C_n^m P_n^m e^{im\phi}r^n$ 
be the local expansion. Then 
\begin{align*}
& \go \mathbf{L} = s_1 + 2\Re{s_2}, \quad \gp \mathbf{L} = s_3 +
  \conj{s_4} \\
& \gp\go \mathbf{L} = s_5+ \conj{s_6}, \quad 
\gp\gp \mathbf{L} = s_7+ \conj{s_8}, \quad 
\go\go \mathbf{L} = s_9 + 2\Re{s_{10}}, 
\end{align*}
where 
\begin{align*}
s_1 & = \sum_{m=1}^p L_n^0 n r^{n-1} P_{n-1}^0 \\
s_2 & = \sum_{n=1}^p \sum_{m=1}^{n-1} L_n^m r^{n-1} P_{n-1}^m
e^{im\phi} \sqrt{\frac{(n+m)!(n-m)!}{(n+m-1)!(n+m-1)!}} \\
s_3 & = \sum_{n=2}^p \sum_{m=0}^{n-2} L_n^m r^{n-1} P_{n-1}^{m+1}
e^{i(m+1)\phi} \sqrt{\frac{(n-m)!}{(n+m)!}} \\
s_4 & = -\sum_{n=1}^p\sum_{m=1}^n L_n^m r^{n-1} P_{n-1}^{m-1}
e^{i(m-1)\phi} \sqrt{\frac{(n+m)!(n -m)!}{(n+m-2)!(n+m-2)!}} \\
s_5 &= \sum_{n=3}^p \sum_{m=0}^{n-3} L_n^m P_{n-2}^{m+1} e^{i(m+1)\phi}r^{n-2} 
\sqrt{\frac{(n-m)!(n+m)!}{(n+m-1)!(n+m-1)!}}\\
s_6 & = \sum_{n=2}^p \sum_{m=1}^{n-1} L_n^m P_{n-2}^{m-1} e^{i(m-1)\phi} r^{n-2}
\sqrt{\frac{(n-m)!(n+m)!}{(n+m-3)!(n+m-3)!}}\\
s_7 &= \sum_{n=4}^p\sum_{m=0}^{n-4} L_n^m P_{n-2}^{m+2} e^{i(m+2)\phi} r^{n-2} 
\sqrt{\frac{(n-m)!}{(n+m)!}}\\
s_8 & = \sum_{n=2}^p \sum_{m=2}^n L_n^m P_{n-2}^{m-2} e^{i(m-2)\phi} r^{n-2} 
\sqrt{\frac{(n+m)!(n-m)!}{(n+m-4)!(n+m-4)!}} \\
& \qquad -
\sum_{n=3}^p L_n^1 P_{n-2}^1 e^{-i\phi} r^{n-2} \sqrt{\frac{(n+1)!}{(n-1)!}}\\
s_9 & = \sum_{n=2}^p L_n^0 n(n-1)  P_{n-2}^0 r^{n-2} \\
s_{10} & = \sum_{n=2}^p \sum_{m=1}^{n-2} L_n^m r^{n-2} P_{n-2}^m e^{im\phi} 
\sqrt{\frac{(n+m)!(n-m)!}{(n+m-2)!(n+m-2)!}}
\end{align*}
\end{theorem}

We point out that there are at least two approaches to computing (\ref{eq2.8})
at each target bead $i$. In the first approach, for each Laplace potential, one
differentiates the multipole or local expansion, evaluates the potential,
gradients, and Hessian at $\bx$, and accumulates the result. In the second
approach, one simplifies each component in (\ref{eq2.8}) into a single expansion
and then carries out the evaluation. Compared with the first approach, the
second approach performs fewer expansion evaluation but consumes more storage to
assemble the final expansion. In RPYFMM, the current implementation adopts the
first approach.

\section{HPX-5 runtime and DASHMM library}
\label{sec:cs} 

\subsection{HPX-5 runtime}
\label{subsec:runtime} 
HPX-5 (High Performance ParalleX) is an experimental Asynchronous Multi-Tasking
(AMT) programming model and runtime developed at Indiana
University~\cite{kulkarni16, kissel16}. Its design is governed by the ParalleX
exascale execution model~\cite{Cimini:2016ab} and it aims to enable programs to
run unmodified on systems from a single SMP to large clusters and supercomputers
with thousands of nodes.

HPX-5 defines a broad API that covers most aspects of the system. Programs are
organized as diffusive, message driven computation, consisting of a large number
of lightweight threads and active messages, executing within the context of a
global address space, and synchronizing through the use of lightweight
synchronization objects. The HPX-5 runtime is responsible for managing global
allocation, address resolution, data and control dependence, and scheduling
threads and the network.

The HPX-5 \emph{global address space} provides a global shared memory space
abstraction and serves as the basis for computation. Global allocation is
performed through a set of dynamic allocators that provide individual, block
cyclic, and user-defined allocation for blocks of memory. Access to data in the
global address space is provided through an asynchronous \texttt{memput/memget}
API. Explicit global address translation can be performed in order to operate on
local machine virtual aliases. Finally, raw global addresses serve as the
targets for HPX-5's active message \emph{parcels}, described below. Localities
(roughly equivalent to MPI processes) are mapped into the global address space
and can be accessed through indices allowing messages to target localities as in
other active message runtimes.

Parcels form the basis of parallel computation in HPX-5. They contain a
description of the action to be performed, argument data, and continuation
information  and are sent to the global address
on which the action is to be performed. The HPX-5 scheduler invokes parcels as
lightweight threads once they reach their destination. This parcel--thread
equivalence is key to abstracting the difference between shared and distributed
execution in HPX-5. Sending a parcel is equivalent to, and the only means of,
spawning a lightweight thread. In shared memory execution it just happens that
all target addresses are on a single locality. Unlike many other AMT runtimes,
HPX-5 is designed around cooperative threading and not simply run to completion
tasks. 

Program data and control dependencies are represented in memory by \emph{local
  control objects} (LCOs). An LCO is an event-driven, lightweight, globally
addressable synchronization object that co-locates data and control
information. All LCOs have input slots, predicates that evaluate functions of
the inputs and may determine that an LCO has been triggered, and continuations
(i.e., dependent threads and parcels) that will be executed once the LCO is
triggered. This allows the user to build fully dynamic dataflow networks managed
by the runtime. A simple example of an LCO is a reduction that performs a sum
across its inputs. HPX-5 is delivered with a number of classes of built-in LCOs,
e.g., futures and reduction types, and permits user-defined LCO classes as
well. A user-defined LCO encodes the data that it represents, the task performed
when an input becomes available, and the predicate under which the LCO is
considered to be triggered.

\subsection{DASHMM} 
\label{subsec:dashmm} 

The Dynamic Adaptive System for Hierarchical Multipole Methods is an open-source
scientific software library built on top of HPX-5 runtime system that aims to
provide an easy-to-use system that can provide scalable, efficient, and unified
evaluations of general HMMs on both shared and distributed memory
architectures. Unlike conventional practice in many existing MPI+X
implementations, which use static partitioning of the global tree structure and
bulk-synchronous communication of the locally essential tree \cite{Warren1992,
  Ying2003, Kurzak2005a}, DASHMM considers the distribution of the directed
acyclic execution graph, represented implicitly as a network of expansion
objects. By leveraging the LCO construct in HPX-5, the expansion object not only
encloses data, but also encodes dependency and continuation. The execution of
DASHMM is therefore completely data-driven and asynchronous. 

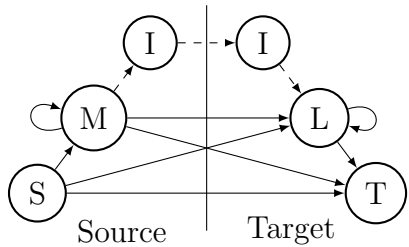
\begin{figure} 
\centering
\begin{tikzpicture}[>=latex, xscale=0.75]
\node at (1.5, -0.5) {Source};
\node at (4.5, -0.5) {Target}; 

\begin{scope}[nodes={draw, thick}]
\draw (3, -0.5) to (3, 2.5); 
\node at (0,0) [circle] (s) {S}; 
\node at (1,1) [circle] (m) {M}; 
\node at (2,2) [circle] (is) {I}; 
\node at (4,2) [circle] (it) {I}; 
\node at (5,1) [circle] (l) {L}; 
\node at (6,0) [circle] (t) {T}; 
\end{scope}

\draw[->] (s) to (t); 
\draw[->] (s) to (m); 
\draw[->] (s) to (l); 
\draw[->] (m) to (l); 
\draw[->] (l) to (t); 
\draw[->] (m) to (t); 
\draw[->, loop left] (m) to (m); 
\draw[->, loop right] (l) to (l); 

\draw[->, dashed] (m) to (is); 
\draw[->, dashed] (is) to (it); 
\draw[->, dashed] (it) to (l); 
\end{tikzpicture} 
\caption{Operator diagrams in DASHMM. Basic multipole methods use
  multipole (M) and local (L) expansions, and eight operators (shown
  in solid lines) that connect them to the sources (S) and targets
  (T). Advanced multipole methods use intermediate expansions (I) and
  three additional operators (shown in dashed lines). The M-to-L
  operation is decomposed into a chain of M-to-I, I-to-I, and I-to-L
  operations in advanced multipole methods.} 
\label{fig:op} 
\end{figure} 

DASHMM can be used for a wide variety of applications, and it does so with an
easy-to-use and general interface. Further, DASHMM's interface is HPX-5
oblivious, meaning that no knowledge of HPX-5 is required to use DASHMM in
either its basic or advanced forms. The basic interface allows users to rapidly
apply the built in methods and kernels in end-science applications. DASHMM
currently provides three built in multipole methods: Barnes-Hut
\cite{Barnes1986}, the classic fast multipole method~\cite{greengard1987fast,
  Carrier1988}, and a variant of FMM that uses the exponential expansion to
reduce arithmetic complexity~\cite{greengard1997new}.  The last method is
referred to as {\tt FMM97} in DASHMM. DASHMM also provides three built in
kernels: the scaling invariant Laplace potential, the scaling variant Yukawa
potential, and the oscillatory Helmholtz potential in low-frequency regime.
The advanced interface allows users to implement variants of the multipole
method, new kernels, or new distribution policies that guide the placement of
expansion data. 

The simplicity and generality of DASHMM are achieved by adopting two
strategies. First, DASHMM uses C++ templates to accept {\tt SourceData}, {\tt
  TargetData}, {\tt Method}, and {\tt Expansion} as parameters for automatic
generalization over user-specified types. Second, DASHMM extends the Expansion
abstraction from its mathematical counterpart and introduces three additional
operators. Each expansion object in DASHMM is either a normal expansion that
corresponds to the usual multipole/local expansion or an intermediate expansion
that is needed for advanced multipole methods (see
Figure~\ref{fig:op}). Additionally, each expansion object  in DASHMM is a
collection of views, each of which is a mathematical approximation stored in
compact form. Together, they facilitate the implementation of various multipole
and multipole-like methods and concurrent kernel evaluation on the same input
data.

\section{Software Installation and Numerical Examples} 
\label{sec:install}

\subsection{Installation} 
RPYFMM depends on two external libraries: the HPX-5 runtime system,
and the DASHMM library. The current version of RPYFMM depends on
version 4.1.0 of HPX-5 or later, which can be downloaded from
\url{https://hpx.crest.iu.edu/download.}  DASHMM is automatically
downloaded in the RPYFMM when the library is built.

Users must first install HPX-5 on their systems. HPX-5 can be built
without or with network transports. For the latter, HPX-5 currently
specifies two network interfaces: the ISend/IRecv (ISIR) interface
with the MPI transport, and Put-with-Completion (PWC) interface with
the Photon transport. Assume that you have unpacked the HPX-5 source
into the folder {\tt /path/to/hpx} and want to install HPX-5 into {\tt
  /path/to/install}. The following steps build and install HPX-5. 
\begin{verbatim}
> cd /path/to/hpx
> % without network transport
> ./configure --prefix=/path/to/install 
> % with ISIR interface 
> % ./configure --prefix=/path/to/install --enable-mpi 
> % with PWC interface 
> % ./configure --prefix=/path/to/install --enable-mpi --enable-photon
> % ./configure --prefix=/path/to/install --enable-pmi --enable-photon
> make 
> make install
\end{verbatim}
The {\tt --enable-mpi} or {\tt --enable-pmi} option for the PWC network is used
to build support for {\tt mpirun} or {\tt aprun} bootstrapping because HPX-5 
does not provide its own distributed job launcher. Please see the official
documentation for more detailed installation instructions for certain Cray
machines. To finish the setup for HPX-5, one modifies the following environment
variables 

\begin{verbatim}
export PATH=/path/to/install/bin:$PATH
export LD_LIBRARY_PATH=/path/to/install/lib:$LD_LIBRARY_PATH
export PKG_CONFIG_PATH=/path/to/install/lib/pkgconfig:$PKG_CONFIG_PATH
\end{verbatim}

Once HPX-5 is installed, assume the RPYFMM is unpacked in directory {\tt
  /path/to/rpyfmm}, will be built in directory {\tt /path/to/rpyfmm/build}, and
be installed into directory {\tt /path/to/rpyfmm/install}, the library can be
built in the following steps using CMake, version 3.4 or higher 
 
\begin{verbatim}
> cd /path/to/rpyfmm/build
> cmake ../ -DCMAKE_INSTALL_PREFIX=/path/to/rpyfmm/install
> make 
> make install
\end{verbatim}
 
\subsection{Example} 
Included with RPYFMM is a test code that demonstrates a simple use of the
library. This code is given in {\tt /path/to/rpyfmm/demo/.} The demonstration
code is not built/installed by {\tt make install}. To build it, run {\tt make
  demo} in {\tt /path/to/rpyfmm/build/demo/.} User can request a summary of the
options to the demo code by running the code with {\tt --help} as a command line
argument. In the following, we provide some walk through of the demo code. 

The basic usage of RPYFMM is through an {\tt Evaluator} object of the DASHMM 
library. The {\tt Evaluator} object is a template over four types: the source
type, the target type, the expansion type and the method type. For example, 
\begin{verbatim} 
dashmm::Evaluator<Source, Target, 
                  dashmm::RPY, dashmm::FMM97> rpy_eval{} 
\end{verbatim} 
declares an {\tt Evaluator} for the RPY kernel using the advanced FMM method for
two user-defined types {\tt Source} and {\tt Target}. The minimum requirements
for the {\tt Source} and {\tt Target} types are 
\begin{verbatim}
struct Source {
  dashmm::Point position;
  double q[3];            // "charges" of the source 
  //...
}; 

struct Target{
  dashmm::Point position; 
  double value[3];        // store the results
  //...
}; 
\end{verbatim} 
Users can declare the object {\tt Evaluator} object with a single type if that
type satisfies both the minimum requirements for {\tt Source} and {\tt Target}
(see type {\tt Bead} of the demo code). 

There are four parameters associated with the RPY kernel, specifying the radius
of the bead, the Boltzmann constant, the absolute temperature, and the solvent
viscosity. This information needs to be passed to the library by declaring 
\begin{verbatim} 
std::vector<double> kparams(1, radius); 
\end{verbatim} 
which uses the default values for the rest, or 
\begin{verbatim}
std::vector<double> kparams(4);
kparams[0] = raidus; 
kparams[1] = boltzmann_constant;
kparams[2] = temperature;
kparams[3] = viscosity
\end{verbatim} 

Finally, an evaluation of the RPY kernel on a set of source and target points 
can be completed by instantiating the {\tt FMM97} method and calling the {\tt
  evaluate} member function of the {\tt Evaluator} object. 
\begin{verbatim}
dashmm::FMM97<Bead, Bead, dashmm::RPY> fmm97{}; 
err = rpy_eval.evaluate(bead_handle, bead_handle, threshold, 
                        fmm97, accuracy, &kparams); 
\end{verbatim}
where {\tt threshold} specifies the refinement limit (max number of
beads allowed in a childless leaf node). 

\subsection{Numerical Results} 
We demonstrate the performance of the RPYFMM library using the demo code,
particularly focusing on the resulting scalability. The tests were performed on
a Cray XC30 cluster at Indiana University, running Linux kernel
3.0.101-0.47-102. Each compute node has two Intel Xeon
E5-2650 v3 processors at 2.3 GHz clock rate and 64 GB of DDR3 RAM. All compute
nodes are connected through the Cray Aries interconnect. The RPYFMM library and
the demo code were compiled using GNU compiler 6.2.0 with {\tt -O3} 
optimization flags. The configuration of the tests can be summarized as follows:

\begin{itemize} 
\item Two data distributions are tested: (a) Uniform distribution inside a cube
  and (b) Uniform distribution on a spherical surface. 
\item The problem size is 15 million for a cube distribution and 8 million for
  a sphere distribution. 
\item Three, six, and nine digits of accuracy are tested. The refinement limit
  for three, six, and nine digits are 80, 100, and 120, respectively. 
\item Tests requiring three digits of accuracy start from one compute node. 
  Tests with six and nine digits of accuracy start from two compute
  nodes. Tests use up to 512 compute nodes for strong scaling evaluation. 
\item All tests are repeated five times and the average is reported here. 
\end{itemize}

The accuracy results of the tests are given in Table~\ref{tab:rpyfmm}. They were
computed according to formula \cite[Eq.(57)]{Cheng1999} at 400 randomly selected
points. The scaling results of the tests are summarized in
Figure~\ref{fig:rpyfmm}. At $12,288$ cores using 512 compute nodes, RPYFMM is
able to compute both problem within one second at three-digit accuracy
requirement. 

\begin{table} 
\centering
\begin{tabular}{cll}
\hline
& Cube & Sphere \\
\hline
3-digit & 2.1410e-3 & 2.1420e-3\\
\hline
6-digit & 1.4115e-7 & 1.3949e-7 \\
\hline
9-digit & 2.6447e-9 & 2.6347e-9 \\
\hline 
\end{tabular}
\caption{Accuracy results of the RPYFMM.}
\label{tab:rpyfmm}
\end{table}

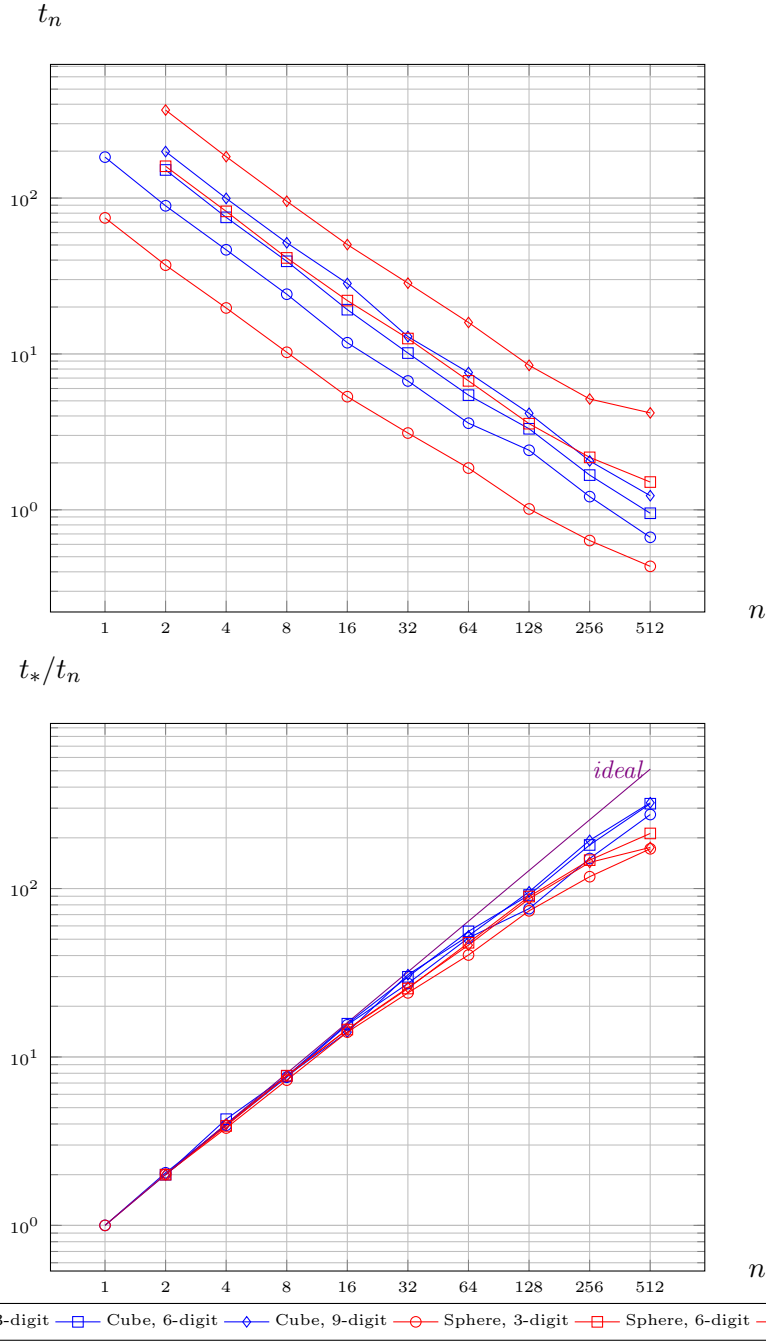
\begin{figure} 
\centering
\begin{tikzpicture}
  \begin{loglogaxis}[     
      ylabel = {\small $t_n$},  
      xlabel = {\small $n$}, 
      every axis x label/.style={at={(ticklabel* cs:1.05)}, anchor=west,}, 
      every axis y label/.style={at={(ticklabel* cs:1.05)}, anchor=south,}, 
      width = 0.75\textwidth, 
      log basis x = 2, 
      xticklabel=\pgfmathparse{2^\tick}\pgfmathprintnumber{\pgfmathresult},
      tick label style={font=\tiny}, 
      grid = both, 
      legend cell align=left, 
      legend entries = {Cube, 3-digit\\
        Cube, 6-digit\\
        Cube, 9-digit\\
        Sphere, 3-digit\\
        Sphere, 6-digit\\
        Sphere, 9-digit\\
      }, 
      legend columns= -1, 
      legend style={font={\tiny}, inner xsep=5pt}, 
      legend to name=scale,
    ]
    \addplot [blue, mark=o] coordinates {
      (1, 183.11200)
      (2, 89.20330)
      (4, 46.57120)
      (8, 24.16110)
      (16, 11.79870)
      (32, 6.70775)
      (64, 3.59531)
      (128, 2.40989)
      (256, 1.21489)
      (512, 0.66593)
    };
    \addplot [blue, mark=square] coordinates {
      (2, 151.47200)
      (4, 75.22500)
      (8, 39.30250)
      (16, 19.23010)
      (32, 10.12820)
      (64, 5.43984)
      (128, 3.31019)
      (256, 1.66979)
      (512, 0.95014)
    }; 
    \addplot [blue, mark=diamond] coordinates {
      (2, 198.8320)
      (4, 99.5396)
      (8, 51.6660)
      (16, 28.3114)
      (32, 12.9156)
      (64, 7.5951)
      (128, 4.1597)
      (256, 2.0579)
      (512, 1.2310)
    }; 
    \addplot [red, mark=o] coordinates {
      (1, 74.64550)
      (2, 37.11930)
      (4, 19.72390)
      (8, 10.24870)
      (16, 5.31695)
      (32, 3.10811)
      (64, 1.85233)
      (128, 1.01258)
      (256, 0.63504)
      (512, 0.43366)
    }; 
    \addplot [red, mark=square] coordinates {
      (2, 159.9720)
      (4, 82.1561)
      (8, 41.3027)
      (16, 21.9918)
      (32, 12.5556)
      (64, 6.7166)
      (128, 3.5695)
      (256, 2.1700)
      (512, 1.5067)
    };
    \addplot [red, mark=diamond] coordinates {
      (2, 367.1)
      (4, 184.4090)
      (8, 95.3989)
      (16, 50.2754)
      (32, 28.4465)
      (64, 15.9185)
      (128, 8.4506)
      (256, 5.1385)
      (512, 4.1848)
    }; 
  \end{loglogaxis}
\end{tikzpicture} 
~
\begin{tikzpicture}
  \begin{loglogaxis}[     
      ylabel = {\small $t_*/t_n$},  
      xlabel = {\small $n$}, 
      every axis x label/.style={at={(ticklabel* cs:1.05)}, anchor=west,}, 
      every axis y label/.style={at={(ticklabel* cs:1.05)}, anchor=south,}, 
      width = 0.75\textwidth, 
      log basis x = 2, 
      xticklabel=\pgfmathparse{2^\tick}\pgfmathprintnumber{\pgfmathresult},
      tick label style={font=\tiny}, 
      grid = both, 
    ]
    \addplot [blue, mark=o] coordinates {
      (1, 1.0000)
      (2, 2.0527)
      (4, 3.9319)
      (8, 7.5788)
      (16, 15.5197)
      (32, 27.2986)
      (64, 50.9308)
      (128, 75.9836)
      (256, 150.7231)
      (512, 274.9731)
    };
    \addplot [blue, mark=square] coordinates {
      (2, 2.0000)
      (4, 4.272)
      (8, 7.7080)
      (16, 15.7536)
      (32, 29.9109)
      (64, 55.6899)
      (128, 91.5186)
      (256, 181.4264)
      (512, 318.8411)
    }; 
    \addplot [blue, mark=diamond] coordinates {
      (2, 2.0000)
      (4, 3.9950)
      (8, 7.6968)
      (16, 14.0461)
      (32, 30.7894)
      (64, 52.3580)
      (128, 95.5990)
      (256, 193.2378)
      (512, 323.0362)
    }; 
    \addplot [red, mark=o] coordinates {
      (1,1)
      (2, 2.0110)
      (4, 3.7845)
      (8, 7.2834)
      (16, 14.039)
      (32, 24.016)
      (64, 40.298)
      (128, 73.718)
      (256, 117.54)
      (512, 172.13)
    }; 
    \addplot [red, mark=square] coordinates {
      (2, 2)
      (4, 3.8943)
      (8, 7.7463)
      (16, 14.548)
      (32, 25.482)
      (64, 47.634)
      (128, 89.633)
      (256, 147.44)
      (512, 212.35)
    };
    \addplot [red, mark=diamond] coordinates {
      (2, 2)
      (4, 3.9814)
      (8, 7.6961)
      (16, 14.604)
      (32, 25.810)
      (64, 46.122)
      (128, 86.882)
      (256, 142.88)
      (512, 175.45)
    }; 
    \addplot [forget plot, samples=20, domain=1:512, color=violet] {x}
    node [left, scale=0.8] {{\it ideal}};
  \end{loglogaxis}
\end{tikzpicture} 
\ref{scale}
\caption{The time to completion $t_n$ (top) and speed relative 
$t_*/t_n$ (bottom) as a function of locality cores $n$.  Each locality has 24
  cores and 64 GB of RAM. For three digits of accuracy, $t_* = t_1$. For six
  and nine digits of accuracy, $t_* =t_2$. The problem size for cube
  distribution is 15 million and the problem size for sphere distribution is 8
  million. }
\label{fig:rpyfmm} 
\end{figure} 

\section{Conclusion}
\label{sec:conclusion}
In this paper, we present RPYFMM, a parallel adaptive fast multipole method
(FMM) software package on shared and distributed memory computers for the
Rotne-Prager-Yamakawa tensor in biomolecular simulations. RPYFMM decomposes the
RPY tensor into a linear combination of Laplace potentials, which are evaluated
using the advanced FMM~\cite{greengard1997new}. RPYFMM is an essential building
block to enable full cell dynamics simulation, which is within reach in the near
term as demonstrated by our numerical results. 

The performance of RPYFMM can be further improved depending on the application
context and underlying architecture. First, future version of DASHMM library
will support heterogeneous memory architecture. As a result, the near field and
possibly other part of the computation can be offloaded to the accelerators.
Second, the FMM for the Laplace potential in RPYFMM is based on spherical
harmonics and exponential expansions. Spherical harmonic expansions are the
orthogonal basis functions for the Laplace operator for a unit sphere and the
number of expansion terms is optimal for spherical geometry. However, compared
with polynomial expansion based approach
\cite{yokota2009fast,yokota2013petascale}, the benefits might not be obvious,
especially at lower accuracy requirement, because polynomials can be evaluated
much more efficiently than spherical harmonics and exponential functions. Future
version of RPYFMM will either have a more tuned evaluations of harmonic and
exponential expansions or will internally switch to a polynomial based
approach when the accuracy requirement is low.

From a broader perspective, authors of this paper are also working on several
closely related research projects in solving (\ref{eq:Ermak}). Some examples
include introducing more realistic HI models by adding the rotational motions of
the beads; accelerating the simulations when rigid body assumptions are valid
(the relative locations of the atoms in portions of the molecules are fixed in
time); developing efficient mesh generation tools for dynamic simnulations;
developing parallel and more efficient time integration schemes; developing
better preconditioners to further reduce the number of iterations in the
iterative schemes when solving the dynamic equations. Research results along
these directions will be discussed in the future.

\section*{Acknowledgments} 
The authors gratefully acknowledge the inspiring discussions with Profs. David
Keyes and Rio Yokota on different parallelization strategies for our solver. 
BZ was supported in part by National Science Foundation grant number
ACI-1440396. GH was supported in part by National Institute of Health
grant number GM 31749.
This research was supported in part by Lilly Endowment, Inc.,
through its support for the Indiana University Pervasive Technology Institute. 
Part of the work was finished when JH was a visiting professor at the King 
Abdullah University of Science and Technology. 

\bibliographystyle{elsarticle-num}
\biboptions{sort&compress}
\bibliography{huangbib}

\end{document}